\documentclass[12pt,onecolumn,draftcls]{IEEEtran} 

\usepackage{verbatim}
\usepackage{amssymb}
\usepackage{amscd}
\usepackage{amsmath}
\usepackage{graphicx}
\usepackage{epstopdf}
\usepackage{epsfig}
\usepackage{amsthm}
\usepackage{url,xcolor}

\newtheorem{remark}{Remark}
\newtheorem{theorem}{Theorem}[section]
\newtheorem{corollary}{Corollary}[section]
\newtheorem{proposition}{Proposition}[section]

\begin{document}
\title{The Control Theory \\ of Motion-Based Communication:\\  Problems in Teaching Robots to Dance}
\author{J.\ Baillieul 
\thanks{ John Baillieul is with the Boston University, email:
johnb@bu.edu.  Support for this work is gratefully acknowledged to ODDR\&E MURI07 Program Grant Number FA9550-07-1-0528, the National Science Foundation ITR Program Grant Number DMI-0330171, and the Office of Naval Research, and by ODDR\&E MURI10 Program Grant Number N00014-10-1-0952.}\& K.\ \"Ozcimder}
\IEEEaftertitletext{\vspace{-2\baselineskip}} 
\maketitle

\begin{abstract}
The paper describes results on two components of a research program focused on {\em motion-based communication} mediated by the dynamics of a control system.  Specifically we are interested in how mobile agents engaged in a shared activity such as dance can use motion as a medium for transmitting certain types of messages.   The first part of the paper adopts  the terminology of {\em motion description languages} and deconstructs an elementary form of the well-known popular dance, Salsa, in terms of four {\em motion primitives} (dance steps).  Several notions of dance complexity are introduced.  We describe an experiment in which ten performances by an actual pair of dancers are evaluated by judges and then compared in terms of proposed complexity metrics.  An energy metric is also defined.  Values of this metric are obtained by summing the lengths of motion segments executed by wheeled robots replicating the movements of the human dancers in each of the ten dance performances.  Of all the metrics that are considered in this experiment, energy is the most closely correlated with the human judges' assessments of performance quality.

The second part of the paper poses a general class of dual objective motion control problems in which a primary objective (artistic execution of a dance step or efficient movement toward a specified terminal state) is combined with a communication objective.  Solutions of varying degrees of explicitness can be given in several classes of problems of communicating through the dynamics of finite dimensional linear control systems.  In this setting it is shown that the cost of adding a communication component to motions that steer a system between prescribed pairs of states is independent of those states.   At the same time, the optimal encoding problem itself is shown to be a problem of packing geometric objects, and it remains open.  Current research is  aimed at solving such communication-through-action problems in the context of the motion control of mobile robots.
\end{abstract}


\section{Introduction}\setcounter{equation}{0}

In \cite{Wong} and \cite{WB} the first author, in collaboration with W.S.\ Wong, studied the concept of {\normalsize\em control  communication complexity} as a formal approach for studying a group of distributed agents exercising independent actions to achieve common goals.  For distributed cooperative systems, it is natural to expect that communication can help improve system performance, and adopting this viewpoint, we report research on the design of control laws whose purpose is to elicit system responses that convey messages to an observer or set of observers.  The idea that the motions of a member of a multiagent team might allow an observer to infer intent or some other meaning is not surprising.  There is, however, no well-developed theory of which we are aware to guide the design motion control laws to enable  mobile agents  to convey useful information to an observer.  Nevertheless, we point to some preliminary efforts in the work reported in \cite{DJ1},\cite{DJ2}, and \cite{DJ3}.  We also note the interesting work of Justh and  Krishnaprasad (\cite{krishna}) which describes the reverse problem of executing motions that cannot be detected by an observer.  

Perhaps the largest body of published research on communication problems involving teams of autonomous robotic agents has treated various aspects of multi-agent consensus ([2]-[6]).  With the exception of [6], this work has modeled communication among agents as occurring instantaneously, and none of the work has provided detailed models of the communication channels.  While it is both practical and convenient in many cases to assume that communication occurs over well understood RF and optical technologies that offer high data rates---allowing the assumptions of instantaneous data exchange and negligible expenditure of energy, our purpose in what follows below is to pose the problem of two or more mobile agents communicating by means of their motions.

This problem is important for a number of reasons.  While motion-based communication will typically have orders of magnitude less information bandwidth than standard RF or free space optical communication, it offers significant stealth potential.  There is perhaps a more important justification for a formal study of this communication mode, however, and that is that in many contexts, communication with observers is inevitable whenever mobile agents execute motions aimed at achieving some objective.  Athletes playing team sports communicate with both their own team mates and with players on the opposing team.  For their own team mates, they of course wish to signal their intentions as clearly and reliably as possible.  For opposing players, however, their goal may well be to communicate so as to make their intended next movement as obscure and unpredictable as possible.  Such non-linguistic communication has been acknowledged by, among others, W.\ Weaver, who wrote that communication involves ``not only  written and oral speech, but also music, the pictorial arts, the theatre, the ballet, and in fact all human behavior.'' (\cite{ShannonWeaver})  In the present paper, dance is the motivating application for our discussion of action-mediated communication.  Section 2 describes a simple form of the popular dance, Salsa.  In Section 3, we describe an experiment in which video recordings of a pair of dancers performing ten basic routines were analyzed and scored by twenty different judges.  Several metrics of performance complexity are proposed.  Section 4 analyzes the dance performances in terms of their deconstruction into eight-bar phrases.  A figure of merit based on energy is also introduced.  In Section 5 we pose the problem of controlling robotic dancers such that they are able to exchange messages through the motions that they execute.  The section explores several classes of problems of communicating through the dynamics of finite dimensional linear control systems in which the systems are operating so as to satisfy terminal endpoint constraints while encoding messages in the system output.   The message encoding problem is shown to be equivalent to a problem of packing geometric objects, and at this writing it remains open.

\section{Salsa --- A prototype problem in control-mediated communication}\setcounter{equation}{0}

The universal set-up in studying (motion) control-mediated communication is a finite set of control or motion primitives in terms of which all communication takes place.  There are many domains of human endeavor in which one can look for motion primitives.  There are, for instance, standard driving protocols for cars entering street intersections, standard approach paths for aircraft landing at airports, and standard movements in competitive athletics.  Among these, it is perhaps most natural to think about motion-mediated communication in the context of dance. Throughout many cultures, dance consists of sequences of body movements that are known to expert dancers, and passed on through formal instruction to beginners and students.  The artistic content of formalized movements that occur in dance is central to what must be expressed in the motion-based language associated with each dance vernacular.  It would seem natural, then, to develop a formal means of transcribing basic motion primitives for dance, but attempts to do this have not led to widespread use among dance professionals.  Perhaps the best known effort in this direction was the development in the 1920 of {\em labanotation}. (\cite{Laban},\cite{Hutchinson}).  Rolf Von Laban attempted to develop a scripting language that was sufficiently expressive that all human movement could be described and recorded on paper.  This has never been widely used, probably because in its attempt to be universally applicable, it became complex and nonintuitive.  (This is supported by the ``more than 700 symbols that indicate parts of the body, direction, levels, and types of movement and the durations of each action.''  (Quoted from the web page \cite{LabanWriter}.)

In this paper, we avoid dealing with such high complexity by restricting our attention to a form of dance involving only a small set of motion primitives---beginner's Salsa.  As is always the case in the performing arts, there are distinct levels of proficiency in Salsa.  Because our goal is to analyze and deconstruct component motions in order to reinterpret them as controlled motions of simple wheeled robots, we consider a version of beginner's Salsa that uses only four basic steps which we label $A,B,C$, and $D$.  (See Appendix A.) In Salsa, as in many forms of dance, each motion primitive (dance step) begins and ends in accordance with the rhythm of the music to which it is set.  Specifically, in a sequence of elementary steps making up a Salsa performance, each of the four motion primitives is executed for a period of eight beats (two bars) of the music.  In what follows, we describe the Salsa motion primitives in terms of their two-bar durations.  We assume each primitive has the dance partners standing in a standard initial pose---illustrated in Figure \ref{fig:dance} and in more detail in Appendix A.

\begin{figure}[htbp] 
   \centering
\includegraphics[width=4in]{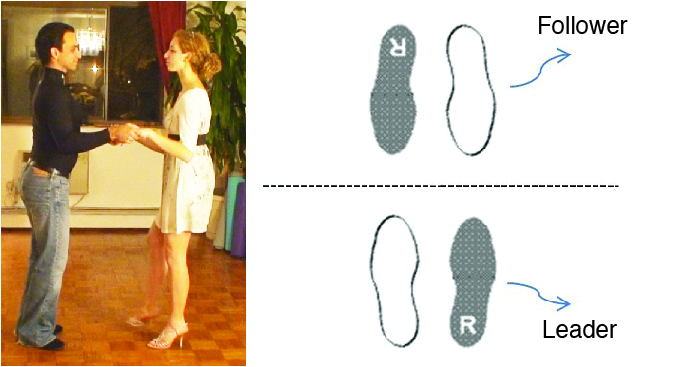}
   \caption{It is assumed that each of the four basic Salsa movements begins and ends with the dance partners facing each other with feet symmetrically placed with respect to a line as depicted.}
   \label{fig:dance}
\end{figure}

\section{The complexity and artistic merit of a dance performance}\setcounter{equation}{0}

In an attempt to understand something about how people perceive the artistic merit of a dance performance, two dancers were asked to perform a number of short Salsa segments using the for basic dance primitives in different sequences.  Digital video recordings of the Salsa segments  were shown to twenty ``judges'' who were asked to rank the performances in order of artistic merit.  The judges included both trained dancers as well as people with no formal training in dance.  All judges were instructed to use standard criteria in their rankings, including artistic conent, dance routine difficulty, partner synchronization, and  complexity of the
choreography.  Ten dance sequences, each comprised of 23 basic dance primitives were  selected to be ranked by each of the judges.  Using the motion primitives (dance steps) $A,B,C,D$ described in the previous section, the ten performances are given in Table 1.

\begin{center}
\footnotesize
\begin{tabular}{||l|c||}
\hline
Dance 1 &  $DDADBBBBACCCDDDDDBDAAAA$\\ \hline
Dance 2 & $AAAAAAAADDDDDDDDDDDBDBB$\\ \hline
Dance 3 & $ADBCDACBDADBCDABACDACBD$\\ \hline
Dance 4 & $DBCADBCADBCADBCADBCADBC$\\ \hline
Dance 5 & $ACBDACBDACBDACBDACBDACB$\\ \hline
Dance 6 & $ABCDBCDACDABDABCABADBCD$\\ \hline
Dance 7 & $DBADACBDDBABDDAACDBBDAD$\\ \hline
Dance 8 & $AAAABAAADAAAAAAACAAADAA$\\ \hline
Dance 9 & $DBCDCBBDCBDDDBDDDAABCCC$\\ \hline
Dance 10 & $DBDCCBDDBBDDDCCCCABDDDB$\\ \hline
\end{tabular}\\[0.2in]
The ten dance sequences.\\
{\bf Table 1}
\end{center}

The average scores of the twenty judges are given in the first row of Table 2. Dance sequence 9 was preferred, while almost no one liked dance number 2.  It is noted that the judges were in substantial agreement regarding dance number 2 (rated as poor) there was comparatively high variance in the judges scores on other dances.


Having thus tabulated the judges' rankings, we were led to the question of whether the artistic qualities in terms of which the performances were differentiated could be indentified in a precise and even quantitative way.  The late Dennis Dutton identified {\em complexity} as one of the four central characteristics of great art\footnote{Dutton identifies the four central characteristics of high art as 1.\ complexity, 2 serious content, 3.\  purpose, and 4. distance. (\cite{dutton})}. (\cite{dutton})
To evaluate the complexity of a sequence of symbols such as those in Table 1, we considered metrics suggested by the well-known Shannon Entropy.  The simplest possible metric may be arrived at by recording the number of occurrences of each of the symbols in the symbol set ${\cal S}=\{A,B,C,D\}$.  Each dance is exactly 23 symbols in length, and thus the relative  frequency of occurrence of the $k$-th symbol is $f_k=(\# {\rm of\ occurrences\ of}\ k-th\ {\rm symbol})/23$.  The metric
\begin{equation}
\left(\begin{array}{c}
{\rm symbol}\\
{\rm frequency}\\
{\rm complexity}\end{array}\right) = -\sum_{k=1}^4 f_k\log_2 f_k
\label{eq:jb:complexity}
\end{equation}
is then a (possibly crude) measure of the variability of the component steps that make up the dance.  Because there are only four symbols involved, the maximum value this measure could take is $\log_2 4=2$, which would be attained if each symbol appeared in the sequence equally often.  (Since the sequence lengths are all 23, this bound is never achieved.)  On the other hand, if any single symbol were to appear in all 23 places in the sequence, the complexity (\ref{eq:jb:complexity}) would have the value 0.  When the complexity metric (\ref{eq:jb:complexity}) is evaluated on the ten dance sequences of Table 1, the values are strictly between the two extremes, and they are given in row three of Table 2.

\vspace{0.1in}
\begin{table*}[ht]
\begin{center}
\footnotesize
\begin{tabular}{||c|cccccccccc||}
\hline
Dance no. & 1 & 2 & 3 & 4 & 5 & 6 & 7 & 8 & 9 & 10\\ \hline\hline
Average&&&&&&&&&&\\[-0.1in]
 score by  & 3.6  & 1.9  & 5.1  & 5.7 & 7.3 & 6 & 6 & 4.2 & 7.8 & 7.3 \\[-0.1in]
judges &(2.0)&(1.5)&(3.0)&(2.1)&(1.8)&(2.6)&(2.2)&(3.1)&(2.4)&(2.2)\\ \hline
Symbol &&&&&&&&&&\\[-0.1in]
frequency & 1.897& 1.403 & 1.985 & 1.996 & 1.996 & 1.996 & 1.848 & 0.927 & 1.848 & 1.731\\[-0.1in]
complexity &&&&&&&&&&\\ \hline
Averaged     &&&&&&&&&&\\[-0.1in]
phrase & 0.625 & 0.162 &  1.8 & 2 & 2 & 1.9 & 1.5 & 0.487 & 1.362 & 1.362\\[-0.1in]
complexity &&&&&&&&&&\\ \hline
Number &&&&&&&&&&\\[-0.1in]
of phrases & 2.322 & 1.522 & 2.322 & 0 & 0&  2.322 & 2.322 & 1.922 & 2.322  & 2.322 \\[-0.1in]
complexity &&&&&&&&&&\\ \hline
Robot &&&&&&&&&&\\[-0.1in]
dance &13727&12945&14326&14567&14547&14248&13349&13181&14627&14647\\[-0.1in]
energy &&&&&&&&&&\\ \hline\hline
\end{tabular}\\[0.2in]
\begin{flushleft}{\footnotesize There were twenty judges;  numbers in parentheses in the average score row are standard deviations.  Robot dance energy is discussed in Section 5.}\end{flushleft}
{\bf Table 2}
\end{center}
\end{table*}

\vspace{0.1in}

A simple linear regression in which the average judges' scores were regressed on the computed symbol frequency suggests only a modest correlation.  (See Fig.\ 2.)  Indeed, the value of the coefficient of correlation for the sequences is only 0.48, indicating a weak correlation.  The following section describes some refined notions of complexity that may more faithfully reflect the artistic quality of the sequences.



\begin{figure}[htbp] 
   \centering
\includegraphics[width=3.5in]{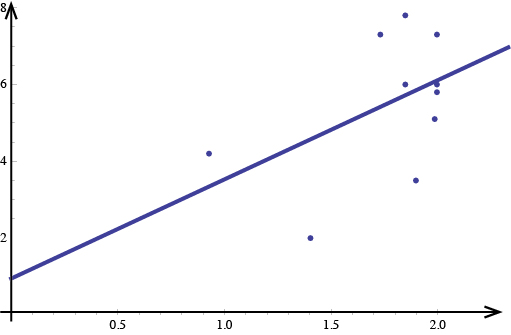}
   \caption{A scatter plot of Judges' rankings as a function of the frequency complexity using the data from Table 1.}
   \label{fig:scatter}
\end{figure}

\section{Deconstructing the dances into four-step phrases}

It is an interesting exercise to attempt to fit four-state Markov chain models to the symbol sequences of Table 1.  While the sequences are long enough and the sets of transitions are rich enough in some cases to construct such models, any model of the dance in which the next step depends only on its immediate predecessor step will probably seem a bit aimless and not reflective of the artistic quality of the sequence of steps that the dance actually contained.  As has been noted in the computer music literature, \cite{Roads}, higher order Markov chain models can be used to capture the phrasal nature of music.  While fitting higher order Markov models to the sequences of Table 1 is beyond the scope of the paper, we shall briefly examine the phrasal structure of the sequences.

\begin{center}
\tiny
\begin{tabular}{||l|c|c||}
\hline\hline
\multicolumn{2}{|c|}{ }&{\# phrases}\\ \hline
Dance 1 &  $(DDAD)(BBBB)(ACCC)(DDDD)(DBDA)AAA$& 5\\ \hline
Dance 2 & $(AAAA)(AAAA)(DDDD)(DDDD)(DDDB)DBB$& 3\\ \hline
Dance 3 & $(ADBC)(DACB)(DADB)(CDAB)(ACDA)CBD$& 5\\ \hline
Dance 4 & $(DBCA)(DBCA)(DBCA)(DBCA)(DBCA)DBC$& 1\\ \hline
Dance 5 & $(ACBD)(ACBD)(ACBD)(ACBD)(ACBD)ACB$& 1\\ \hline
Dance 6 & $(ABCD)(BCDA)(CDAB)(DABC)(ABAD)BCD$& 5\\ \hline
Dance 7 & $(DBAD)(ACBD)(DBAB)(DDAA)(CDBB)DAD$& 5\\ \hline
Dance 8 & $(AAAA)(BAAA)(DAAA)(AAAA)(CAAA)DAA$& 4\\ \hline
Dance 9 & $(DBCD)(CBBD)(CBDD)(DBDD)(DAAB)CCC$& 5\\ \hline
Dance 10 & $(DBDC)(CBDD)(BBDD)(DCCC)(CABD)DDB$& 5\\ \hline
\end{tabular}\\[0.2in]
\footnotesize
\begin{flushleft}{\footnotesize The dance sequences grouped into four-letter phrases.  The right hand column lists the number of distinct four letter phrases in the sequence, and the final three letters in each sequence were not counted as a phrases.}\\
\end{flushleft}
{\bf Table 3}
\end{center}

As noted above, each of the four motion primitives (or dance steps) is executed over a period of eight beats of music.  In Salsa, each phrase is eight musical measures in length.  Since there are four beats to a measure, it is natural to group the letters in the sequences into four letter phrases.  Several phrase centric complexity metrics can then be considered.  One such metric is based on viewing each four symbol phrase as a  complete dance sequence in its own right.  In terms of the symbol set ${\cal S}$, every four letter phrase has a complexity given by (\ref{eq:jb:complexity}) where now $f_k=(\#  $ number of occurrences of the $k$-th symbol)/4.  Clearly, there are five possible values that this phrase complexity metric can take on phrases made up of the four letters in ${\cal S}$.  They are $0,\ -\frac{1}{4}\log\frac{1}{4}-\frac{3}{4}\log\frac{3}{4}=0.811278,\ -\log\frac{1}{2}=1,\ -\frac{1}{2}\log\frac{1}{4}-\frac{1}{2}\log\frac{1}{2}=1.5,$ and $\log 4=2$ in the respective cases that all letters in the phrase are equal, three letters in the phrase are equal, there are two distinct pairs of equal letters, there are exactly three letters in the sequence, and finally in the case that there are four distinct letters in the sequence.  Based on this phrase metric, we prescribe an {\em averaged phrase complexity} metric for each of the twenty-three letter sequences.  Ignoring the final three letters in each sequence, the right hand column in Table 3 lists the number of distinct four letter phrases that make up the dance. The third row of Table 2 lists the averaged phrase complexity of the dance.

A further metric in terms of which to evaluate dance complexity is what we shall call the {\em number-of-phrases complexity}.  We omit details but note that this metric is based on the number of distinct phrases and their frequency of occurrence among the first twenty letters in each dance sequence (a number between 1 and 5).  The possible values of the number-of-phrases complexity in terms of the appropriately restated formula (\ref{eq:jb:complexity}) range between 0 and $\log_2 5\approx 2.344$.  The values taken on by this metric for our ten dances are listed in row 4 on Table 2.  Note that  while dances 4 and 5 have the highest averaged phrase complexities (being comprised of four distinct letters), they also have the lowest complexity measured in terms of number-of-phrases.  

Comparing the average judges' scores are with the averaged phrase complexity showed a discernible correlation, with the coefficient of correlation being 0.75.  On the other hand, the number-of-phrases complexity had no meaningful correlation with the judges rankings (correlation coefficient -0.099).  It is interesting to note, however, that a convex combination of these complexity metrics in which the relative weightings are 90\% averaged-phrase complexity and 10\% number-of-phrases complexity has a slightly higher value of 0.764 coefficient of correlation with the judges rankings.  This metric slightly discounts dance routines that repeat the same four steps over and over.  It is also interesting to note that both these complexity metrics are identical on and do not discriminate between dances 4 and 5, and yet the judges had a clear preference for dance 5.  There is clearly some aspect of artistic merit that is not captured by the complexity metrics.

\section{Control Energy as a Communication Metric}\setcounter{equation}{0}

A major theme in recent work of Wong and Baillieul (\cite{CDC09},\cite{Wong}, and \cite{WB}) is understanding the complexity of communicating through a control system in terms of the required control energy.  As dance requires physical exertion, it seems natural to compare the dance sequences of Section 3 in terms of the amount of energy required to perform them.  While energy data was not recorded for the human dancers who performed the ten dance routines, we have done some simple energy calculations on wheeled robots in our lab doing appropriately stylized versions of the same beginner Salsa routines.  These are tabulated in row 5 of Table 2.  The units are centimeters, and it was assumed that the energy expended was proportional to the amount the robots moved in performing the dance. For these robot replications of the given dances, the coefficient of correlation with the judges' ranking was 0.8.  Performance energy thus emerged as the metric most closely correlated with the rankings of the judges.

Up to this point, we have discussed possible ways to measure the artistic complexity of dance routines.  We now turn to the question of complexity of communication between partners in a dance performance.  The two dancers in a Salsa performance have distinctly different roles: one (usually the male) is the leader, and the other is a follower.  Initiation of motion associated with each step in the dance involves communication between the leader and follower about which of the four component steps they are going to perform next.  This communication occurs through gestures and movement.  Skilled performers blend such motion-based communication seamlessly into the execution of each step, and they are undoubtedly able to communicate effectively without undue expenditure of energy.  Energy takes on greater significance in the case that dance routines are created for pairs of mobile robots that must communicate with each other through motions that involve many fewer physical degrees of freedom than  what are available to human dancers.  The energy needed for reliable motion-based communication will typically be relatively large for a mobile robot that must simultaneously execute the  motion prescribed by the dance step, while at the same time signaling to a partner what the intended next step will be.  Early results on motion based communication in which wheeled  robots are controlled to move so as to achieve a primary motion objective while simultaneously transmitting a motion-encoded message have been reported in \cite{DJ1},\cite{DJ2}, and \cite{DJ3}.


\begin{figure}[htbp] 
   \centering
\includegraphics[width=3.5in]{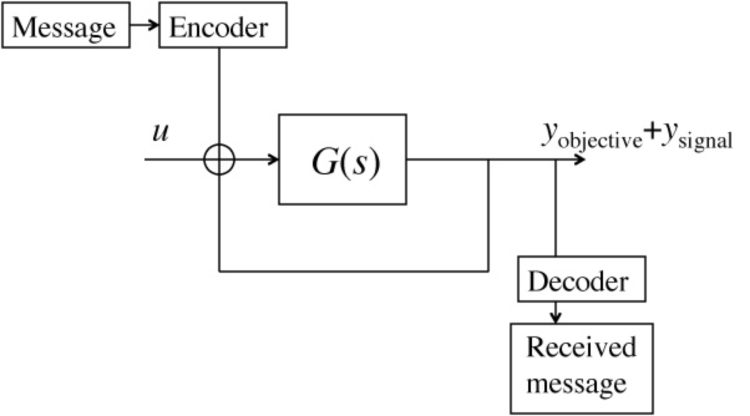}
   \caption{The primary control objective provides the {\em context} in which an overlaid message is transmitted.}
   \label{fig:ControlComm}
\end{figure}

In the theory of formal languages, notions of {\em context} are of interest in discussing both the reliability of accurate reception and the possibility of compression of a transmitted message.  Context has an even greater importance in motion-based communication.  Gestures inevitably have meanings that crucially depend on the context in which they are made.  Beyond this, the ability to transmit a meaningful motion-based signal while an agent engaged is carrying out a prescribed motion toward a given objective may depend on the nature of that objective.  A wheeled robot vehicle, for example,  may communicate by means of a short duration sinusoidal deviation from a prescribed path, but its ability to execute such a deviation will be limited if it is moving close to an obstacle.  Our goal is to understand not only how motion based signaling may depend on context, and in particular how the energy cost of communicating through the dynamics of a control system will depend on the context in which communication takes place.  In the spirit of \cite{CDC09}, we study the energy cost of communicating through a control system by formulating an optimal control problem.  The specific case of wheeled robots will be considered elsewhere, but to fix ideas about the nature of such optimization problems, we consider a problem of controlling a finite dimensional linear system with a finite number (say $m$) of inputs such that all inputs steer the system to a specified terminal state in one unit of time following paths that can be distinguished from one another. Thus each distinct path followed by the system output will correspond to a transmitted message. The precise statement is as follows:

\noindent{\bf Energy-optimal control communication through a finite-dimensional linear system:}
Find a finite set of control inputs $u_j(\cdot)$, $j=1,\dots,m$, each of which  steers
 the state of
\begin{equation}
\begin{array}{l}
\dot x=Ax+Bu\\
y=Cx
\end{array}
\label{eq:jb:basicLin}
\end{equation}
from $x(0)=0$ to $x(1)=x_1\in\mathbb{R}^n$ in one unit of time such that the corresponding output functions $y_1(\cdot),\dots,y_m(\cdot)$ satisfy
\begin{equation}
\Vert y_i-y_j\Vert_{*}\ge \epsilon>0,
\label{eq:jb:outputCrit}
\end{equation}
 for $ i\ne j$, and  such that
\begin{equation}
\eta=\int_0^1\, \sum_{j=1}^m\Vert u_j(t)\Vert^2\,dt
\label{eq:jb:Cost}
\end{equation}
is minimized.  We take the norm appearing in the control cost $\eta$ to be the standard finite dimensional Euclidean norm.  Several function space norms are  of interest for gaging the separation of the output trajectories in (\ref{eq:jb:outputCrit}), but for the purpose of the present paper, we consider the standard $L_2$ norm.

In the case that $m=1$, the requirement of (\ref{eq:jb:outputCrit}) is vacuous, and the problem reduces to a classical liner-quadratic optimal control problem.  We note that the solution to the problem as posed will have a strong dependence on the specified endpoint $x_1$.  This represents a {\em communication context}, changes in which could have a significant effect on both the cost and the solution to the communication problem.  This explicit incorporation of {\em context} in our formulation of the problem is an important step in understanding how communication  and the cost of communication between agents will depend on the respective states in which a transmitting agent ($=$ dance leader) and receiving agent ($=$ dance follower) find themselves.
For the linear control problem, we shall show that while changes in $x_1$ affect the cost, the way in which messages are optimally encoded in the motion is independent of $x_1$.
The optimal control problem (\ref{eq:jb:basicLin})-(\ref{eq:jb:Cost}) may be abstracted as follows.  

\noindent{\bf Problem A: Multiple distinguishable solutions to a linear algebra problem:}
Let $U,V,W$ be real vector spaces such that $\dim U\gg\dim V$,and $\dim U\sim\dim W$.  Let $\ell:U\to W$ and $L:U\to V$ be linear operators such that $L$ has full rank $=\dim V$.  For a given vector $x\in V$, find $m$ solutions $u_1,\dots,u_m\in U$ to the equation
\begin{equation}
Lu=x
\label{eq:jb:abstract}
\end{equation}
such that the objective function
\begin{equation}
\eta=\sum_{j=1}^m\Vert u_j\Vert^2
\end{equation}
is minimized subject to the constraint that
\begin{equation}
\Vert\ell u_i-\ell u_j\Vert = \epsilon>0
\end{equation}
for $i\ne j$.

\begin{theorem}
Suppose $W=U$, and let $\ell = I$ be the identity operator.  Assume that $\dim U-\dim V\ge m-1$.  Let ${\cal N}(L)$ denote the null-space of $L$ and let
$S = \{u\in {\cal N}(L)\ :\ \Vert u\Vert =\left( {\frac{m-1}{2m}}\right)^{\frac{1}{2}} \epsilon\}$.  Finally, let $\vec n_1,\dots,\vec n_m\in S$ be the vertices of any $m-1$-simplex.  Any solution to Problem A is of the form
\[
u_j=u_0+\vec n_j,\ j=1,\dots,m
\]
for any such  choice of $\vec n_j\in S$,  
and $u_0=L^T(LL^T)^{-1}x$.  Moreover, the optimal value of the objective function is
\[
m\,x^T(LL^T)^{-1}x + \frac{m-1}{2}\epsilon^2.
\]
\end{theorem}
\begin{proof}
Adjoin the problem constraints to the objective function in the usual way with Lagrange multipliers to get the modified objective function
\[
\sum_{j=1}^m\Vert u_j\Vert^2 + \sum_{j=1}^m \lambda_j^T(Lu_j-x) + \sum_{i<j}\mu_{ij}(\Vert u_i-u_j\Vert^2-\epsilon^2).
\]
The critical point equations are written
\[
2u_j+L^T\lambda_j+2\sum_{i\ne j}\mu_{ij}(u_j-u_i)=0,\ \ j=1,\dots,m.
\]
Multiplying all equations by $L$, we find that all $\lambda_j=-2(LL^T)^{-1}x$.  Plugging this into the critical point equations and summing, we obtain $\sum_{j=1}^m u_j=m\,L^T(LL^T)^{-1}x$.  Note that we can write $u_j=u_0+\vec n_j$ where $u_0=L^T(LL^T)^{-1}x$ and $\vec n_j\in{\cal N}(L),\ j=1,\dots,m$.  

Since $\vec n_j\perp u_0$, it is easy to see that minimizing $\sum \Vert u_j\Vert^2$ is equivalent to minimizing 
\[
\sum_{j=1}^m\Vert u_0+\vec n_j\Vert^2= n\,\Vert u_0\Vert^2 + \sum_{j=1}^m\Vert\vec n_j\Vert^2
\]
subject to $\Vert \vec n_i - \vec n_j\Vert = \epsilon$ for all $i\ne j$.  Since $u_0$ is fixed, the problem is thus equivalent to minimizing $\sum_{j=1}^m\,\Vert\vec n_j\Vert^2$ subject to $\Vert \vec n_i-\vec n_j\Vert = \epsilon$ for $i\ne j$.  These constraints specify that $\vec n_1,\dots \vec n_m$ are vertices of an $m-1$-simplex centered at $0\in U$ and having diameter $\epsilon$.  It was shown in \cite{blumenthal} that the sphere of smallest radius $r$ that contains the vertices $\vec n_1,\dots,\vec n_m$ has
\[
r=\left(\frac{m-1}{m}\right)^{\frac{1}{2}}\epsilon.
\]
Hence, the optimal choices of $\vec n_j$ have 
\[
\Vert \vec n_j\Vert^2=\frac{m-1}{2m}\epsilon^2,
\]
and therefore the optimizing solution to the problem yields a minimum value of
\[
m\,\Vert u_0\Vert^2+m \Vert \vec n_j\Vert^2=
m\,x^T(LL^T)^{-1}x +\frac{m-1}{2}\epsilon^2.
\]
\end{proof}

While the solution to the optimization problem given in this theorem highlights the general characteristics of solutions to the energy optimal communication problem specified by (\ref{eq:jb:basicLin})-(\ref{eq:jb:outputCrit}), there are no comparably explicit expressions for the solution to the control problems.  The current state of knowledge in the case of integrator systems is given by the following.

\begin{proposition}
Consider positive integers $n\ge m$.  The energy-optimal problem of finding $m$ distinguishable solutions steering an $n$-th order integrator, $x^{(n)}=u,\, y=x$ between fixed endpoints $x(0)=x^{\prime}(0)=\cdots=x^{(n-1)}(0)=0$ and $x(1)=x_1,x^{'}(1)=x_2\cdots,x^{(n-1)}(1)=x_n$ is equivalent to a problem of finding least squares optimal solutions with separation constraints to a finite dimension system of linear equations.  Specifically, let $N\ge m+n-2$.  There is a one-one equivalence between solutions to the energy-optimal communication problem for the state space representation of the $n$-th order integrator and the problem of finding solutions $\vec a_1,\dots,\vec a_m\in\mathbb{R}^N$ to 
\[
L\vec a=\vec x =(x_1,\dots,x_n)^T
\]
($\vec x$ the terminal endpoint of the integrator problem) such that the objective function
\[
\eta=\sum_{j=1}^n \vec a_j^TQ\,\vec a_j
\]
is minimized subject to the distinguishability constraint
\[
(\vec a_i-\vec a_j)^TR\,(\vec a_i-\vec a_j)=\epsilon.
\]
$Q$ and $R$ are $(N+1)\times (N+1)$ symmetric positive definite matrices and $L:\mathbb{R}^{N+1}\to\mathbb{R}^n$.  These are given explicitly in the constructive proof.\label{MainProposition}
\end{proposition}
\begin{proof}
The state space rendering on the integrator system is obtained by writing $x_j(t)=x^{(j-1)}(t),\ j=1,\dots,n$.  Then, since the problem assumes $x_j(0)=0$, we can immediately express $x_j$ in terms of the control $u$:
\[
x_j(t)=\int_0^t\frac{(t-s)^{n-j}}{(n-j)!}u(s)\,ds.
\]
Consider controls that are polynomials: $u(s)=a_0+a_1s+\cdots+a_Ns^N$.  Then an elementary but lengthy calculation yields
\begin{equation}
x_j(1)=\frac{1}{(n-j)!}\sum_{k=1}^{N+1}\sum_{i=1}^{n-j}(-1)^i{{n-j}\choose{i}}\frac{1}{k+i}a_{k-1}.
\label{eq:jb:operatorL}
\end{equation}
From (\ref{eq:jb:operatorL}) one can immediately write an expression for each entry in a matrix representation of the operator $L$.

A similar calculation shows that
\[
\int_0^1(\vec a\cdot t^{[N]})^2\,dt = \vec a^TQ\,\vec a
\]
where $\vec a = (a_0,a_1,\dots.a_N)^T$, and $t^{[N]}=(1,t,t^2,\dots,t^N)^T$, and $Q$ is the $(N+1)\times(N+1)$ symmetric matrix whose $ij$-th entry is $\frac{1}{i+j-1}$.

With somewhat more effort, we can develop an explicit expression for $R$.  We note that in the expansion of the expression
\[
\int_0^1\left[(1-s)^{n-1}(b_0+b_1s+\cdots + b_Ns^N)\right]^2\,ds,
\]
the term in which the product $b_kb_{\ell}$ appears is
\[
b_kb_{\ell}\int_0^1\,\sum_{i=0}^{n-1}\sum_{j=0}^{n-1}(-1)^{i+j}{{n-1}\choose{i}}{{n-1}\choose j}s^{i+j+k+\ell}\,ds.
\]
Carrying out the elementary integration step, and using combinatorial identities, one can show that this quantity is equal to
\[
b_kb_{\ell}\frac{(2(n-1))!(k+\ell)!}{(2n-1+k+\ell)!}.
\]
This defines the $k\ell$-element of the symmetric matrix $R$ and completes the proof of the proposition.
\end{proof}

This result shows in an explicit way that the energy-optimal control communication problem for the $n$-th order integrator can be restated as a distinguishable-solutions to linear equations problem: With $L,Q,R$ as described in the proof of Proposition \ref{MainProposition}, find $m$ solutions $\vec a_1,\dots,\vec a_m$ to
\[
L\vec a = \vec x
\]
such that the objective function
\[
\eta=\sum_{j=1}^m\vec a_j^TQ\,\vec a_j
\]
is minimized subject to the constraints
\[
(\vec a_i-\vec a_j)^TR\,(\vec a_i-\vec a_j)=\epsilon,\ (i\ne j).
\]

The problem can be solved very much along the lines of the proof of Theorem 5.1.  We briefly sketch this in order to point out common features as well as an important difference.  First, note that the solution $\vec a_0$ to $L\vec a=\vec x$ which minimizes $\vec a^TQ\,\vec a$ is given by the weighted pseudo inverse:
\[
\vec a_0 = Q^{-1}L^T(LQ^{-1}L^T)^{-1}\vec x.
\]
Following exactly the same reasoning as in the earlier proof, the $m$ solution that minimize $\eta$ while satisfying the distinguishability constraint are of the form
\[
\vec a_j=Q^{-1}L^T(LQ^{-1}L^T)^{-1}\vec x + \vec n_j,\  j=1,\dots,m
\]
where the $\vec n_j$ lie in the null space of $L$ and minimize $\sum_{j=1}^m\vec n_j^TQ\,\vec n_j$ subject to $(\vec n_i-\vec n_j)^TR\,(\vec n_i-\vec n_j)=\epsilon$.
The minimal value of the objective function for the problem is then given by
\begin{equation}
\eta_0 = m\vec x^T(LQ^{-1}L^T)^{-1}\vec x + \sum_{j=1}^m \vec n_j^TQ\vec n_j.
\label{eq:jb:OptCost}
\end{equation}

\begin{remark}
For the problem as stated, we are unaware of a solution of closed form simplicity comparable to what was given in Theorem 5.1.  The way in which the value of $\eta_0$ depends on the norms used to specify the distinguishability constraint is the subject of future research.
\end{remark}
\begin{remark}
The cost component of the optimal $\eta_0$ that is due to $\vec x$ is independent of the ``message'' components of the solution and depends only on the magnitude of $\vec a_0 = Q^{-1}L^T(LQ^{-1}L^T)^{-1}\vec x$.  While this {\em separation property} will apparently hold in general for energy-optimal control communication problems posed for linear systems, we do not expect such a clean separation between a primary objective and the cost of overlaying a motion-encoded message in nonlinear systems.  Understanding control communication in nonlinear problems is a topic of current research. 
\end{remark}


\begin{corollary}
For each $N\ge n$, let $L_{N+1}:\mathbb {R}^{(N+1)}\to\mathbb{R}^n$ and $Q_{N+1}=\left(\frac{1}{i+j-1}\right)_{i,,j=1}^{N+1}$ be as constructed in the proof of Proposition 5.1.  Then
\begin{description}
\item (i) $Q_{N+1}^{-1}L_{N+1}^T = \left[\begin{array}{c}
Q_{n}^{-1}L_{n}^T\\
\cdots \\[-0.14in]
 0_{(N-n+1)\times n}
\end{array}\right]$ where $0_{(N-n+1)\times n}$ is an $(N-n+1)\times n$ submatrix all of whose entries are $0$, and 
\item (ii) Taking the domain of the mapping $L_{N+1}$ to be the space $\mathbb{R}_N[s]$ of all real polynomials of degree $\le N$, and defining the projection operator
\[
\begin{array}{c}
P_{N+1} = \\
I_{N+1}-Q_{N+1}^{-1}L_{N+1}^T(L_{N+1}Q_{N+1}^{-1}L_{N+1}^T)^{-1}L_{N+1}
\end{array}
\]
we see that $P_{N+1}$ is a projection of $\mathbb{R}_N[s]$ onto the polynomial subspace spanned by shifted Legendre polynomials (defined on $[0,1]$) of degrees $k$, where $n\le k\le N$.
\end{description}
\end{corollary}
\begin{flushright}
$\Box$
\end{flushright}
We omit the proof, but note that this provides transparent demonstration that the least squares solution to the stated quadratic optimal control problem for the $n$-th order integrator restricted to polynomial control inputs is a polynomial of degree $n-1$.

\begin{remark}
The result of Proposition 5.1 is that if $N-n\ge m-2$, $m$ symbols may be encoded by control trajectories of an $n$-th order integrator with quadratic optimal cost given by (\ref{eq:jb:OptCost}) where we have chosen $\vec n_1,\dots,\vec n_m$ such that they solve the following associate subproblem:

\noindent{\bf Integrator Message Encoding Subproblem:} Given symmetric positive definite matrices $Q_{N+1}$ and $R_{N+1}$ as defined in the proof of Prop.\ 5.1, find $m$ vectors $\vec n_j\in{\cal N}(L_{N+1}$ such that
\[
\sum_{j=1}^m\vec n_j^TQ_{N+1}\,\vec n_j
\] 
is minimized subject to the separation constraint
\[
(\vec n_i-\vec n_j)^TR_{N+1}\,(\vec n_i-\vec n_j)=\epsilon
\]
holding for all $i\ne j$.

At present, this subproblem appears to be quite challenging, and the way in which its solution depends on $m$ and $N$ (for fixed $n$) is an open question.  
\end{remark}

\section{Conclusion}\setcounter{equation}{0}
This paper has described work to deconstruct a simple dance form (beginners' Salsa) to create a motion library from which dance routines can be choreographed and analyzed.  Various figures of merit including {\em phrase complexity} and {\em total expended energy} were applied to ten actual dance routines performed by a pair of dancers.  Rankings of the performances based on these metrics were compared with the rankings of twenty judges.  The metric that was most closely correlated with the average ranking of the judges was expended energy.  Motivated by this observation as well as our own prior work using energy as a figure of merit in the study of control communication complexity (\cite{WB},\cite{WB1}), we formulated a problem of minimum energy communication through a finite dimensional linear control system.

More specifically, we have treated the problem of finding $m$ control inputs $u_1,\dots,u_m$, all of which steer the state of a finite dimensional control system (\ref{eq:jb:basicLin}) between prescribed endpoints so as to minimize the {\em average control energy} (\ref{eq:jb:Cost}) subject to satisfying a distinguishability constraint (\ref{eq:jb:outputCrit}).  We call this the {\em problem of optimal communication through the nullspace of a linear system}.  The extent to which it informs a theory of communication using the controlled motions of mobile robots engaged in group activities such as dance is the focus of current research.  The problem is of interest in its own right, however.  We showed that the problem of steering the state of an $n$-th order integrator $x^{(n)}=u$ from the origin ($x^{(j)}(0)=0, j=0\dots,n-1$) to a prescribed endpoint $x(1)=x_1,x^{\prime}(1)=x_2,\dots,x^{(n-1)}(1)=x_n$ in such way that the average control energy is minimized subject to an output distinguishability constraint
\[
\int_0^1 \left(y_i(t)-y_j(t)\right)^2\,dt=\epsilon^2\ \ {\rm for}\ i\ne j
\] 
can be restated as the problem of finding $m$ polynomials $p_j(s)=\vec a_j\cdots^{[N}]$ of degree $N\ge n$ such that the coefficient vectors $\vec a_j$ minimize the value of a quadratic form $\vec a^TQ_{N+1}\,\vec a$ subject to the separation constraint $(\vec a_i-\vec a_j)^TR_{N+1}(\vec a_i-\vec a_j)=\epsilon^2$.  The solution to this problem and the way that it depends on the values of $m$ and $N$ remains an open question. We conjecture that nullspace communication problems formulated for other classes of finite dimensional linear systems will have similar restatements in terms of appropriate classes of special functions and solutions to problems of optimal packing of geometric objects.

\newpage
\appendices
\section{}\setcounter{equation}{0}
\setcounter{section}{8}

\begin{figure}[htbp] 
   \centering
\includegraphics[width=3.7in]{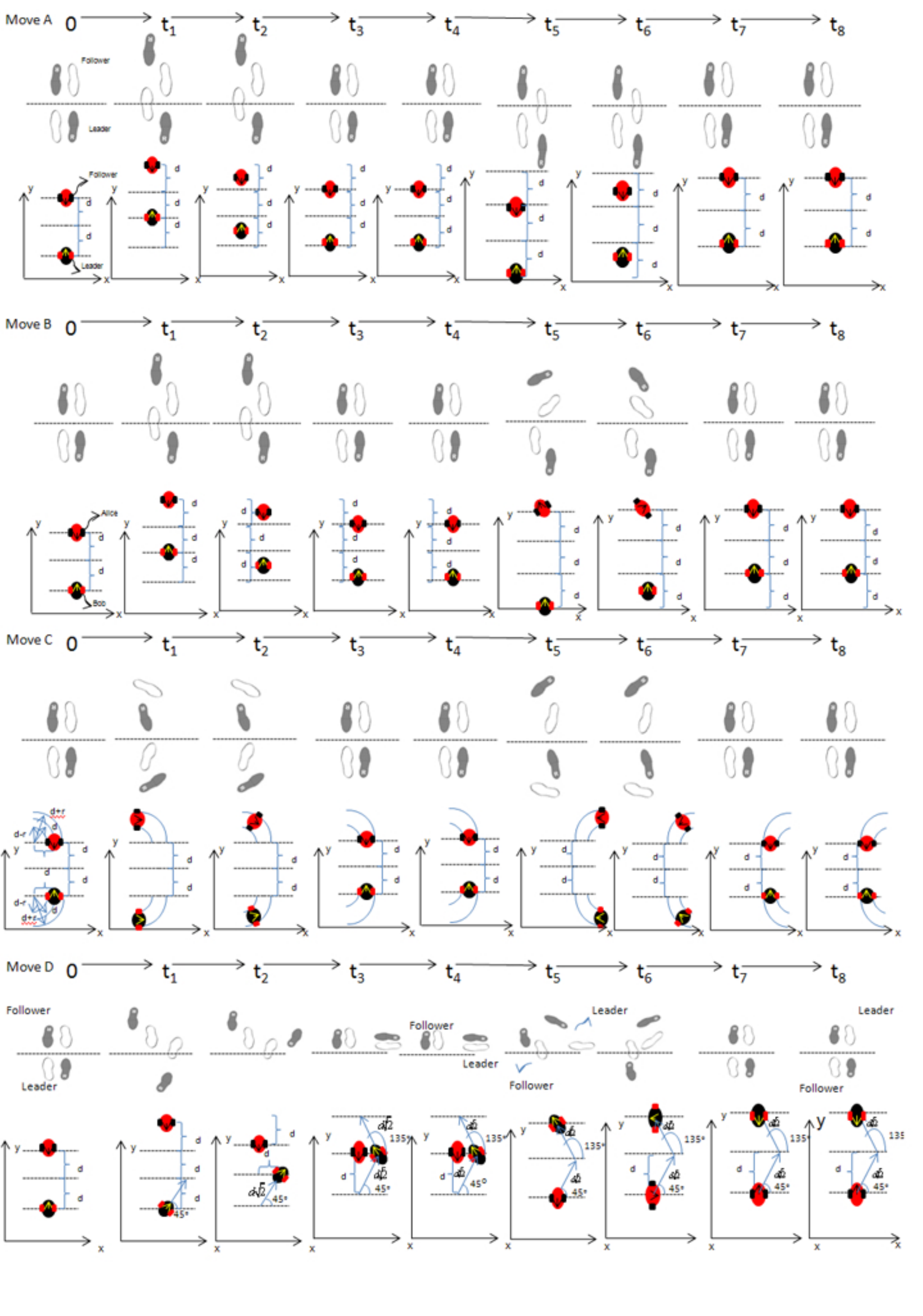}
   \caption{Basic Salsa---four dance steps: from top to bottom A,B,C,D}
   \label{fig:StepsAB}
\end{figure}

\end{document}